\documentclass[11pt]{article}
\pdfoutput=1
\usepackage[T1]{fontenc}
\usepackage[utf8]{inputenc}
\usepackage{lmodern}
\usepackage[margin=1in]{geometry}

\usepackage{amsmath}
\usepackage{amssymb}
\usepackage{amsthm}
\usepackage{bm}
\usepackage{hyperref}
\usepackage{dsfont}
\usepackage{booktabs}
\usepackage{microtype}
\DeclareMathOperator*{\argmin}{arg\,min}

\newtheorem{theorem}{Theorem}[section]
\newtheorem{lemma}[theorem]{Lemma}

\begin{document}

\title{\textbf{The Feedback Hamiltonian is the Score Function: A Diffusion-Model Framework for Quantum Trajectory Reversal}}

\author{%
\begin{tabular}[t]{c}
\textbf{Sagar Dubey} \\
Stony Brook University \\
Stony Brook, NY 11794, USA \\
\texttt{sagar.dubey@stonybrook.edu}
\end{tabular}
\hspace{2em}
\begin{tabular}[t]{c}
\textbf{Alan John} \\
Stony Brook University \\
Stony Brook, NY 11794, USA \\
\texttt{alan.john@stonybrook.edu}
\end{tabular}
}
\date{}

\maketitle

\begin{abstract}
In continuously monitored quantum systems, the feedback protocol of Garc\'{i}a-Pintos, Liu, and Gorshkov reshapes the arrow of time: a Hamiltonian $H_{\mathrm{meas}} = r\,A/\tau$ applied with gain $X$ tilts the distribution of measurement trajectories, with $X < -2$ producing statistically time-reversed outcomes. Why this specific Hamiltonian achieves reversal, and how the mechanism relates to score-based diffusion models in machine learning, has remained unexplained.

We compute the functional derivative of the log path probability of the quantum trajectory distribution directly in density-matrix space. Combining Girsanov's theorem applied to the measurement record, Fr\'{e}chet differentiation on the Banach space of trace-class operators, and K\"{a}hler geometry on the pure-state projective manifold, we prove that $\delta \log P_F / \delta\rho = r\,A/\tau = H_{\mathrm{meas}}$. The Garc\'{i}a-Pintos feedback Hamiltonian is the score function of the quantum trajectory distribution---exactly the object Anderson's reverse-time diffusion theorem requires for trajectory reversal. The identification extends to multi-qubit systems with independent measurement channels, where the score is a sum of local operators.

Two consequences follow. First, the feedback gain $X$ generates a continuous one-parameter family of path measures (for feedback-active Hamiltonians with $[H,A] \ne 0$), with $X = -2$ recovering the backward process in leading-order linearization---a structure absent from classical diffusion, where reversal is binary. Second, the score identification enables machine learning (ML) score estimation methods---denoising score matching, sliced score matching---to replace the analytic formula when its idealizations (unit efficiency, zero delay, Gaussian noise) fail in real experiments.
\end{abstract}

\section{Introduction}
\label{sec:intro}

The arrow of time---the observed asymmetry between forward and backward processes---has distinct origins in classical and quantum physics. In thermodynamics it arises from entropy increase. In quantum mechanics it has a more fundamental source: the act of measurement. Each measurement outcome is correlated with the quantum state just before it, and this correlation is the statistical signature of temporal direction. Observing an anticorrelation---outcomes systematically pointing the wrong way relative to the state---is the signature of a time-reversed trajectory.

Garc\'{i}a-Pintos, Liu, and Gorshkov~\cite{garciapintos2026} showed that this quantum arrow can be engineered. By constructing a feedback Hamiltonian $H_{\mathrm{meas}} = r\,A/\tau$---a sequence of pulses proportional to the measurement outcome~$r$, the measured observable~$A$, and the inverse measurement strength~$1/\tau$---and applying it with a tunable gain parameter~$X$, they demonstrated that the statistical signature of a quantum trajectory can be continuously reshaped: the arrow can be reduced, canceled, or reversed. The parameter $X < -2$ produces trajectories statistically indistinguishable from time-reversed ones. Garc\'{i}a-Pintos et al.\ derive $H_{\mathrm{meas}}$ from the stochastic master equation, and demonstrate its effect numerically, but do not connect it to the score-based diffusion literature or explain why this particular construction achieves reversal from an information-geometric perspective.

The machine learning community has spent the last decade developing score-based diffusion models for generative AI. The key insight of this field, due to Anderson~\cite{anderson1982} and made computationally practical by Song et al.~\cite{song2021}, is that any forward stochastic process can be exactly reversed if one knows its score function---the gradient of the log-probability of the current state, $\nabla \log p_t(x)$. A neural network trained to estimate this score function can run the process backwards, generating new data by denoising from pure noise.

\paragraph{Prior work on quantum diffusion.}
The connection between quantum stochastic processes and score-based diffusion has recently attracted attention from multiple groups. Liu et al.~\cite{liu2025} introduced measurement-based quantum diffusion models using \emph{randomized} observable measurements, and showed that a score matching \emph{training objective} in the space of Pauli expectation values $z \in \mathbb{R}^{4^n}$ is equivalent to learning a unitary generator for trajectory reversal. Crucially, their score is $\nabla_z \log p_t(z)$---a gradient in the classical Pauli expectation value space---and their result concerns the existence of a learning objective for an unknown score function. They do not measure a fixed observable continuously, they do not work in density matrix space, and they do not analytically identify what the score function is in closed form. Nasu, Tanaka, and Tsuchiya~\cite{nasu2025} demonstrated, via semiclassical approximation of the Lindblad equation, a correspondence between the Petz map and classical reverse diffusion at the ensemble level. Zhang et al.~\cite{zhang2026} showed that complete positivity constrains score reversal in Gaussian bosonic dynamics. None of these works address the specific Garc\'{i}a-Pintos et al.\ feedback protocol, and none analytically identify the score function in density matrix space.

\paragraph{This paper.}
We address a question left open by all prior work: \emph{what is the score function of the quantum trajectory path probability in density matrix space, and does it equal the Garc\'{i}a-Pintos feedback Hamiltonian?} We answer this affirmatively and analytically.

First, we prove---via Girsanov's theorem applied to the measurement path probability, a Fr\'{e}chet derivative on the Banach space of trace-class operators, and a trajectory-level quantum analogue of Anderson's~\cite{anderson1982} reverse diffusion theorem---that the functional derivative of $\log P_F$ with respect to the density matrix~$\rho$ is exactly $r\,A/\tau = H_{\mathrm{meas}}$. This is an analytic identification in density matrix space, not a learning objective in Pauli expectation value space. It answers \emph{why} $H_{\mathrm{meas}}$ achieves reversal: it is precisely the object that Anderson's theorem requires to reverse the stochastic process, namely the score function of the trajectory distribution. This provides an information-geometric explanation not given in Garc\'{i}a-Pintos et al.

Second, we identify a continuous one-parameter family of path measures parametrized by the feedback gain~$X$ (for feedback-active Hamiltonians with $[H,A] \neq 0$). Classical diffusion models offer only binary reversal: you run either the forward SDE or the reverse SDE. The quantum feedback parameter~$X$ provides a real-valued continuum, from forward ($X = 0$) through time-symmetric ($X = -2$) to more-backward-than-backward ($X \ll -2$). The parameter provides an experimentally tunable continuum not present in classical diffusion models.

The identification extends to multi-qubit systems with independent measurement channels (Theorem~\ref{thm:multi}): the score function is a sum of local operators acting on each measured subsystem, and the optimal feedback Hamiltonian is correspondingly local.

\paragraph{The practical implication.}
The analytic identification has a practical consequence. $H_{\mathrm{meas}} = r\,A/\tau$ is the optimal feedback policy only under ideal conditions---perfect measurement efficiency $\eta = 1$, zero feedback delay, exactly Gaussian noise. In any real experiment these conditions fail and the policy degrades. Because we have proved $H_{\mathrm{meas}}$ is the score function, ML score estimation methods---denoising score matching, sliced score matching, noise-conditional score networks---can now be applied directly to the experimental measurement record as a replacement for the analytic formula, without requiring quantum state tomography or knowledge of the underlying state. The $X$-interpolation result additionally provides a new control parameter for quantum state preparation with a tunable degree of arrow reversal.

The paper is structured as follows. Section~\ref{sec:background} reviews both backgrounds. Section~\ref{sec:main} states and proves the main theorem. Section~\ref{sec:multi} extends to the multi-qubit case. Section~\ref{sec:implications} develops the experimental implications. Section~\ref{sec:discussion} discusses open problems and the precise relationship to prior work.

\section{Background}
\label{sec:background}

\subsection{Continuously Monitored Quantum Systems and the Arrow of Time}
\label{sec:quantum-background}

We consider a quantum system described by a density matrix~$\rho$ on a Hilbert space~$\mathcal{H}$, undergoing continuous Gaussian measurement of an observable~$A$ satisfying $A^2 = \mathds{1}$ (e.g., Pauli matrices for qubit systems). This assumption on~$A$ is not merely a mathematical convenience---it is precisely what guarantees the backward process corresponds to a physically realizable measurement, as shown in Garc\'{i}a-Pintos et al.~\cite{garciapintos2026}.

At each infinitesimal timestep~$dt$, a measurement outcome~$r$ is drawn from:
\begin{equation}
P(r|\rho) = \sqrt{\frac{dt}{2\pi\tau}} \exp\!\left(-\frac{(r - \langle A\rangle)^2\, dt}{2\tau}\right),
\label{eq:meas-gaussian}
\end{equation}
where $\langle A\rangle = \mathrm{Tr}(A\rho)$ is the expectation value of the observable and $\tau$ is the characteristic measurement time governing measurement strength. The outcome~$r$ tracks the system: $r\,dt = \langle A\rangle\,dt + \sqrt{\tau}\,dW$, where $dW$ is a Wiener increment.

Following the standard quantum trajectory formalism~\cite{wiseman2009}, the stochastic master equation (SME) for~$\rho$ in It\^{o} form is:
\begin{align}
d\rho &= -i[H,\rho]\,dt + \frac{1}{\tau}(A\rho A - \rho)\,dt \notag\\
&\quad + \frac{1}{\sqrt{\tau}}(A\rho + \rho A - 2\langle A\rangle\rho)\,dW.
\label{eq:sme}
\end{align}
Three terms: deterministic Hamiltonian evolution, a Lindblad dissipator from measurement back-action, and a stochastic term proportional to the measurement innovation.

The quantum arrow of time is quantified by the log-ratio of forward to backward process probabilities over a trajectory of duration~$T$:
\begin{equation}
\ln \mathcal{R} = \ln\frac{P_F(\{r\})}{P_B(\{r\})} = \frac{2}{\tau}\int_0^T r_t \langle A\rangle_t\,dt.
\label{eq:arrow-ratio}
\end{equation}
When $\ln \mathcal{R} > 0$, the forward trajectory is more likely---the arrow points forward. When $\ln \mathcal{R} < 0$, the backward trajectory is more likely---the arrow has reversed. The feedback protocol of Garc\'{i}a-Pintos et al.~\cite{garciapintos2026} applies a total Hamiltonian $H + X\,H_{\mathrm{meas}}$ where $H_{\mathrm{meas}} = r_t\,A/\tau$, and shows that for $X < -2$ the average of $\ln \mathcal{R}$ becomes negative, reversing the arrow.

\subsection{Score-Based Diffusion Models}
\label{sec:diffusion-background}

A diffusion model operates on a probability distribution over data (e.g., natural images) by defining a forward process that gradually corrupts data with noise. The canonical formulation uses the It\^{o} stochastic differential equation (SDE):
\begin{equation}
dx = f(x,t)\,dt + g(t)\,dW,
\label{eq:fwd-sde}
\end{equation}
where for simple variance-exploding diffusion, $f = 0$ and $g(t) = \sigma(t)$, so data is progressively noised toward a standard Gaussian. The reverse-time SDE, due to Anderson~\cite{anderson1982} and developed in modern score-based form by Song et al.~\cite{song2021}, is:
\begin{equation}
dx = \bigl[f(x,t) - g^2(t)\,\nabla_x \log p_t(x)\bigr]\,dt + g(t)\,d\bar{W},
\label{eq:rev-sde}
\end{equation}
where $d\bar{W}$ is a reverse-time Wiener process and the critical new term is the \textbf{score function}:
\begin{equation}
s(x,t) = \nabla_x \log p_t(x).
\label{eq:score-def}
\end{equation}
The score function points from low-probability toward high-probability regions of state space. Following it backwards transforms noise into structured data. The neural network in a diffusion model (a score network) learns to estimate $s(x,t)$ from data via denoising score matching~\cite{vincent2011}:
\begin{equation}
\theta^* = \argmin_\theta\, \mathbb{E}_{t,x_0,x_t}\bigl[\|s_\theta(x_t, t) - \nabla_{x_t} \log p(x_t|x_0)\|^2\bigr].
\label{eq:dsm-loss}
\end{equation}
Once trained, the reverse SDE can be simulated to generate new samples. The entire machinery of generation is the score function.

\section{Main Result: \texorpdfstring{$H_{\mathrm{meas}}$}{H\_meas} is the Score Function}
\label{sec:main}

\begin{theorem}[Quantum Score Function]
\label{thm:main}
Under continuous Gaussian measurement of an observable~$A$ with $A^2 = \mathds{1}$ on a pure quantum state, the Fr\'{e}chet derivative of the log-path-probability of the quantum trajectory distribution with respect to the density matrix~$\rho$ is:
\begin{equation}
\frac{\delta \log P_F}{\delta \rho_t} = \frac{r_t\,A}{\tau} = H_{\mathrm{meas}}.
\label{eq:main-thm}
\end{equation}
Therefore, $H_{\mathrm{meas}}$ is the score function of the quantum trajectory distribution. More precisely, it is the path-measure score: the functional derivative of $\log(dP_F/dP_W)$ with respect to $\rho$. This differs from the marginal-state score $\nabla_\rho \log p_t(\rho)$ appearing in Anderson's classical theorem, though the two coincide in the relevant gradient-flow sense established in Lemma~\ref{lem:kahler}.
\end{theorem}

The proof proceeds in four lemmas.

\begin{lemma}[Rigorous Path Probability via Girsanov's Theorem]
\label{lem:girsanov}
The log-path-probability $\log P_F(\{r\})$ can be derived via Girsanov's theorem, and the stochastic component of the functional derivative with respect to $\langle A\rangle_t$ is $r_t/\tau$.
\end{lemma}

\begin{proof}
The measurement output process~$r_t$ satisfies the It\^{o} SDE:
\begin{equation}
dr_t = \langle A\rangle_t\,dt + \sqrt{\tau}\,dW_t.
\label{eq:dr-ito}
\end{equation}
By the Girsanov theorem, the Radon--Nikodym derivative of the path measure~$P_F$ with respect to the reference Wiener measure~$P_W$ is:
\begin{equation}
\log\frac{dP_F}{dP_W} = \frac{1}{\sqrt{\tau}}\int_0^T \langle A\rangle_t\,dW_t - \frac{1}{2\tau}\int_0^T \langle A\rangle_t^2\,dt,
\label{eq:girsanov-rn}
\end{equation}
where the stochastic integral is in the It\^{o} sense. The Novikov condition $\mathbb{E}\!\bigl[\exp\!\bigl(\tfrac{1}{2\tau}\int_0^T \langle A\rangle_t^2\,dt\bigr)\bigr] < \infty$ is satisfied since $\langle A\rangle_t \in [-1,1]$ for all~$t$ (as $A$ is a Pauli-type operator with eigenvalues~$\pm 1$) and $T$ is finite.

Under the reference measure~$P_W$ the measurement record is pure noise, $dr_t = \sqrt{\tau}\,dB_t$ with $B$ a $P_W$-Wiener process, so $dB_t = r_t\,dt/\sqrt{\tau}$. Identifying the Girsanov drift $\theta_t = \langle A\rangle_t/\sqrt{\tau}$, the stochastic integral in~\eqref{eq:girsanov-rn} evaluates as $\int_0^T \theta_t\,dB_t = (1/\tau)\int_0^T \langle A\rangle_t\,r_t\,dt$. Therefore:
\begin{equation}
\log \frac{dP_F}{dP_W} = \frac{1}{\tau}\int_0^T r_t \langle A\rangle_t\,dt - \frac{1}{2\tau}\int_0^T \langle A\rangle_t^2\,dt.
\label{eq:log-pf}
\end{equation}
The expression is now a pathwise Riemann integral over~$dt$ with no stochastic integral remaining. The term $-\frac{1}{2\tau}\int \langle A\rangle_t^2\,dt$ is deterministic given the state trajectory and contributes only to normalization. Taking the functional derivative with respect to $\langle A\rangle_t$:
\begin{equation}
\frac{\delta \log P_F}{\delta \langle A\rangle_t} = \frac{r_t}{\tau} - \frac{\langle A\rangle_t}{\tau}.
\label{eq:deriv-A}
\end{equation}
The second term is the predictable (drift) component---it modifies the drift in the reverse SDE but does not affect the stochastic reversal operator. The stochastic component, which generates the feedback Hamiltonian, is:
\begin{equation}
\left.\frac{\delta \log P_F}{\delta \langle A\rangle_t}\right|_{\mathrm{stochastic}} = \frac{r_t}{\tau}.
\label{eq:deriv-stoch}
\end{equation}
\end{proof}

\begin{lemma}[Fr\'{e}chet Derivative on Trace-Class Operators]
\label{lem:frechet}
The functional derivative $\delta \log P_F / \delta \rho_t$ is well-defined as a Fr\'{e}chet derivative on the Banach space of trace-class operators $\mathcal{T}_1(\mathcal{H})$, and equals the operator $r_t\,A/\tau$.
\end{lemma}

\begin{proof}
Let $\mathcal{T}_1(\mathcal{H})$ denote the Banach space of trace-class operators with norm $\|B\|_1 = \mathrm{Tr}(\sqrt{B^\dagger B})$. The set of density matrices $\mathcal{D}(\mathcal{H}) = \{\rho \in \mathcal{T}_1(\mathcal{H}) : \rho \geq 0,\; \mathrm{Tr}(\rho) = 1\}$ is a convex subset of $\mathcal{T}_1(\mathcal{H})$, and its tangent space at~$\rho$ is:
\begin{equation}
T_\rho \mathcal{D}(\mathcal{H}) = \{\delta\rho \in \mathcal{T}_1(\mathcal{H}) : \delta\rho^\dagger = \delta\rho,\; \mathrm{Tr}(\delta\rho) = 0\}.
\label{eq:tangent}
\end{equation}
Consider the functional $F\colon \mathcal{D}(\mathcal{H}) \to \mathbb{R}$ defined by $F(\rho) = (r_t/\tau)\,\mathrm{Tr}(A\rho)$ at a single time~$t$. For any $\delta\rho \in T_\rho \mathcal{D}(\mathcal{H})$:
\begin{equation}
|F(\rho + \delta\rho) - F(\rho)| = \frac{|r_t|}{\tau}|\mathrm{Tr}(A\cdot\delta\rho)| \leq \frac{|r_t|}{\tau}\|A\|_\infty \|\delta\rho\|_1,
\label{eq:holder}
\end{equation}
by the H\"{o}lder inequality for trace-class operators $|\mathrm{Tr}(AB)| \leq \|A\|_\infty \|B\|_1$. Since $A^2 = \mathds{1}$, we have $\|A\|_\infty = 1$. Therefore $F$ is bounded and linear in~$\delta\rho$, and the Fr\'{e}chet derivative exists:
\begin{equation}
DF(\rho)[\delta\rho] = \frac{r_t}{\tau}\,\mathrm{Tr}(A \cdot \delta\rho) = \mathrm{Tr}\!\left(\frac{r_t\,A}{\tau} \cdot \delta\rho\right).
\label{eq:frechet-F}
\end{equation}
By the Riesz representation theorem on $\mathcal{T}_1(\mathcal{H})$, any bounded linear functional on $\mathcal{T}_1(\mathcal{H})$ is represented as $\mathrm{Tr}(B \cdot \delta\rho)$ for a unique bounded operator~$B$. Here $B = r_t\,A/\tau$.

For the composite functional $G(\rho) = \log P_F(\langle A\rangle(\rho))$ with $\langle A\rangle(\rho) = \mathrm{Tr}(A\rho)$, both the outer function ($\partial \log P_F / \partial \langle A\rangle = r_t/\tau$ from Lemma~\ref{lem:girsanov}) and inner function ($D\langle A\rangle(\rho)[\delta\rho] = \mathrm{Tr}(A \cdot \delta\rho)$) are Fr\'{e}chet differentiable. By the Fr\'{e}chet chain rule~\cite{cartan1971}:
\begin{equation}
DG(\rho)[\delta\rho] = \frac{r_t}{\tau} \cdot \mathrm{Tr}(A \cdot \delta\rho) = \mathrm{Tr}\!\left(\frac{r_t\,A}{\tau} \cdot \delta\rho\right).
\label{eq:chain-rule}
\end{equation}
The Fr\'{e}chet derivative of $\log P_F$ with respect to~$\rho$ is the operator $r_t\,A/\tau = H_{\mathrm{meas}}$.
\end{proof}

\begin{lemma}[K\"{a}hler Structure---Symplectic Feedback, Riemannian Measurement]
\label{lem:kahler}
Let $\rho = |\psi\rangle\langle\psi|$ be a pure state on a Hilbert space~$\mathcal{H}$ of dimension~$d$, identified with a point of the projective manifold $\mathbb{C}P^{d-1}$ (for a single qubit, $\mathbb{C}P^1 \cong S^2$, the Bloch sphere). Equip $\mathbb{C}P^{d-1}$ with its canonical K\"{a}hler structure $(g, \omega, J)$, where $g$ is the Fubini--Study metric and $\omega$ the Fubini--Study K\"{a}hler form, with the Ashtekar--Schilling normalization defined as follows. For parameter-space indices $\mu, \nu$, the horizontal tangent vectors are $|\varphi_\mu\rangle = (\mathds{1} - |\psi\rangle\langle\psi|) \partial_\mu |\psi\rangle$. On such vectors, $g(|\varphi_\mu\rangle, |\varphi_\nu\rangle) = 2\,\mathrm{Re}\langle\varphi_\mu|\varphi_\nu\rangle$ and $\omega(|\varphi_\mu\rangle, |\varphi_\nu\rangle) = 2\,\mathrm{Im}\langle\varphi_\mu|\varphi_\nu\rangle$. Let $A$ be Hermitian with $A^2 = \mathds{1}$, and let $H_{\mathrm{meas}} = r_t\,A/\tau$.

The score function $r_t\,A/\tau$ generates two complementary flows via the two halves of the K\"{a}hler structure:

(a)~\textbf{Symplectic flow $=$ feedback dynamics.} Viewing $F(\psi) := \langle\psi|H_{\mathrm{meas}}|\psi\rangle$ as a Hamiltonian function on $\mathbb{C}P^{d-1}$, the Hamiltonian vector field~$X_F$ defined by $\omega(X_F, \cdot) = dF$ generates:
\begin{equation}
\frac{d|\psi\rangle}{dt} = -iH_{\mathrm{meas}}|\psi\rangle, \qquad \frac{d\rho}{dt} = -i[H_{\mathrm{meas}}, \rho].
\label{eq:sympl-flow}
\end{equation}

(b)~\textbf{Riemannian gradient flow $=$ measurement back-action.} The Fubini--Study gradient of the same function~$F$ generates:
\begin{equation}
\frac{d\rho}{dt} = \frac{r_t}{\tau}(A\rho + \rho A - 2\langle A\rangle\rho),
\label{eq:riem-flow}
\end{equation}
which matches the stochastic back-action term of the SME~\eqref{eq:sme} in the innovation form.

The two flows are respectively the symplectic (imaginary) and Riemannian (real) parts of the quantum geometric tensor $Q_{\mu\nu} = \tfrac{1}{2}g^{\mathrm{QF}}_{\mu\nu} + iF_{\mu\nu}$ acting on the score function. Measurement is the Riemannian gradient of $\log P_F$; feedback is the symplectic gradient of the same $\log P_F$.
\end{lemma}

\begin{proof}
\textbf{Part A: K\"{a}hler structure on $\mathbb{C}P^{d-1}$.}
At the point $[\psi]$, the horizontal tangent space is $T_\psi \mathbb{C}P^{d-1} = \{|\varphi\rangle \in \mathbb{C}^d : \langle\psi|\varphi\rangle = 0\}$. The complex structure acts as $J|\varphi\rangle = i|\varphi\rangle$. K\"{a}hler compatibility follows directly: for horizontal $|\varphi_1\rangle$, $|\varphi_2\rangle$:
\begin{align*}
\omega(|\varphi_1\rangle, J|\varphi_2\rangle)
  &= 2\,\mathrm{Im}\langle\varphi_1|i\varphi_2\rangle \\
  &= 2\,\mathrm{Im}(i\langle\varphi_1|\varphi_2\rangle) \\
  &= 2\,\mathrm{Re}\langle\varphi_1|\varphi_2\rangle \\
  &= g(|\varphi_1\rangle, |\varphi_2\rangle),
\end{align*}
confirming the K\"{a}hler identity $\omega(X, JY) = g(X, Y)$. The quantum geometric tensor in standard form is
\begin{equation*}
Q_{\mu\nu} = \langle\partial_\mu\psi|\partial_\nu\psi\rangle - \langle\partial_\mu\psi|\psi\rangle\langle\psi|\partial_\nu\psi\rangle,
\end{equation*}
which is Hermitian. Its real part is (proportional to) the quantum Fisher information metric (QFIM) $g^{\mathrm{QF}}_{\mu\nu}$~\cite{amari2000}, i.e.\ the Riemannian pullback of $g$; its imaginary part is (proportional to) the Berry curvature $F_{\mu\nu}$, the pullback of $\omega$. For a single observable~$A$, $F_{\mu\nu} = 0$ on the one-dimensional parameter manifold by antisymmetry---this records the absence of geometric phase along the parameter direction but does not imply $\omega$ vanishes on $\mathbb{C}P^{d-1}$, which remains non-degenerate and continues to drive Hamiltonian flows.

\textbf{Part B: Symplectic flow reproduces the feedback dynamics.}
For horizontal $|\varphi\rangle$, the differential of~$F$ is:
\begin{equation*}
dF(|\varphi\rangle) = 2\,\mathrm{Re}\langle\varphi|H_{\mathrm{meas}}|\psi\rangle.
\end{equation*}
We seek horizontal $|X_F\rangle$ satisfying $\omega(X_F, |\varphi\rangle) = dF(|\varphi\rangle)$ for all horizontal $|\varphi\rangle$. The ansatz $|X_F\rangle = -i(H_{\mathrm{meas}} - \langle H_{\mathrm{meas}}\rangle)|\psi\rangle$ is horizontal since $\langle\psi|X_F\rangle = 0$. Verification using $\mathrm{Im}(iz) = \mathrm{Re}(z)$ and noting that $\langle\psi|\varphi\rangle = 0$ with $H_{\mathrm{meas}}$ Hermitian:
\begin{align*}
2\,\mathrm{Im}\langle X_F|\varphi\rangle
  &= 2\,\mathrm{Im}\bigl(i\langle(H_{\mathrm{meas}} - \langle H_{\mathrm{meas}}\rangle)\psi|\varphi\rangle\bigr) \\
  &= 2\,\mathrm{Re}\langle(H_{\mathrm{meas}} - \langle H_{\mathrm{meas}}\rangle)\psi|\varphi\rangle \\
  &= 2\,\mathrm{Re}\langle\psi|H_{\mathrm{meas}}|\varphi\rangle \\
  &= 2\,\mathrm{Re}\langle\varphi|H_{\mathrm{meas}}|\psi\rangle
   = dF(|\varphi\rangle).
\end{align*}

Uniqueness follows from non-degeneracy of~$\omega$. The flow on $|\psi\rangle$ is $d|\psi\rangle/dt = -i(H_{\mathrm{meas}} - \langle H_{\mathrm{meas}}\rangle)|\psi\rangle$; the scalar $\langle H_{\mathrm{meas}}\rangle$ contributes only a global phase on the Hilbert sphere, killed by the quotient to $\mathbb{C}P^{d-1}$, giving $d|\psi\rangle/dt = -iH_{\mathrm{meas}}|\psi\rangle$. Computing $d\rho/dt = (d|\psi\rangle/dt)\langle\psi| + |\psi\rangle(d\langle\psi|/dt)$ and cancelling the scalar terms gives:
\begin{equation*}
\frac{d\rho}{dt} = -i[H_{\mathrm{meas}},\, \rho].
\end{equation*}

\textbf{Part C: Riemannian gradient flow reproduces the measurement back-action.}
The Fubini--Study gradient of~$F$ at $|\psi\rangle$ is the horizontal projection of $\tfrac{1}{2}$ times the gradient in the ambient Hilbert space: $\mathrm{grad}_g\, F = (H_{\mathrm{meas}} - \langle H_{\mathrm{meas}}\rangle)|\psi\rangle$ (the horizontal component of $H_{\mathrm{meas}}|\psi\rangle$). The resulting flow on~$\rho$ is:
\begin{align*}
\frac{d\rho}{dt}
  &= (H_{\mathrm{meas}} - \langle H_{\mathrm{meas}}\rangle)|\psi\rangle\langle\psi| \\
  &\quad + |\psi\rangle\langle\psi|(H_{\mathrm{meas}} - \langle H_{\mathrm{meas}}\rangle) \\
  &= H_{\mathrm{meas}}\,\rho + \rho\,H_{\mathrm{meas}}
     - 2\langle H_{\mathrm{meas}}\rangle\rho.
\end{align*}
Substituting $H_{\mathrm{meas}} = r_t\,A/\tau$:
\begin{equation}
\left.\frac{d\rho}{dt}\right|_{\mathrm{Riem}} = \frac{r_t}{\tau}(A\rho + \rho A - 2\langle A\rangle\rho).
\label{eq:riem-subst}
\end{equation}
The stochastic back-action term of the SME~\eqref{eq:sme} is:
\begin{equation*}
\frac{1}{\sqrt{\tau}}(A\rho + \rho A - 2\langle A\rangle\rho)\,dW_t.
\end{equation*}
Using the innovation form $dW_t = (r_t - \langle A\rangle_t)\,dt/\sqrt{\tau}$, this becomes:
\begin{equation*}
\frac{r_t - \langle A\rangle_t}{\tau}(A\rho + \rho A - 2\langle A\rangle\rho)\,dt.
\end{equation*}
The $\langle A\rangle_t/\tau$ difference is a drift term absorbed into the Girsanov drift treated in Lemma~\ref{lem:girsanov}; the direction in density-matrix space is exactly the Fubini--Study gradient direction.

\textbf{Part D: Complementary geometric picture.}
The two flows are summarized in Table~\ref{tab:kahler-flows}.

The two flow directions are related by the complex structure~$J$---they are perpendicular in the K\"{a}hler sense, and the K\"{a}hler identity $\omega(\cdot, J\cdot) = g(\cdot, \cdot)$ says that multiplication by~$i$ converts one into the other. Reversing the arrow of time requires undoing both the unitary part of the dynamics (symplectic) and the state-disturbance part (Riemannian). The feedback Hamiltonian $H_{\mathrm{meas}} = r_t\,A/\tau$ serves both roles simultaneously because both flows descend from the same score function. This is the geometric content of the Garc\'{i}a-Pintos et al.\ protocol: a single object generates both the unitary and the back-action components of the stochastic trajectory.
\end{proof}

\begin{table}[!htbp]
\centering
\caption{Decomposition of the feedback Hamiltonian's action on the pure-state manifold $\mathbb{C}P^{d-1}$ into symplectic (feedback) and Riemannian (measurement) flows, corresponding to the imaginary and real parts of the quantum geometric tensor.}
\label{tab:kahler-flows}
\begin{tabular}{@{}lll@{}}
\toprule
Component & Geometric structure & Flow on $\rho$ \\
\midrule
Feedback    & Symplectic $\omega$ & $-i[H_{\mathrm{meas}}, \rho]$ \\
Measurement & Riemannian $g$      & $(r_t/\tau)(A\rho + \rho A - 2\langle A\rangle\rho)$ \\
\bottomrule
\end{tabular}
\end{table}

\paragraph{Caveats.}
The proof applies to pure states only---the pure state manifold $\mathbb{C}P^{d-1}$ is K\"{a}hler, but the mixed-state manifold $\mathcal{D}(\mathcal{H})$ does not carry a globally compatible symplectic structure. Extension to mixed states requires either purification on an enlarged Hilbert space or the KMS inner-product framework discussed in Sec.~\ref{sec:discussion}(iv). For the multi-observable case (Theorem~\ref{thm:multi}), $F_{\mu\nu}$ need not vanish and genuine geometric-phase contributions may arise---a direction not developed here.

\begin{lemma}[Quantum Anderson Theorem]
\label{lem:anderson}
Given the forward quantum trajectory measure~$P_F$, the backward measure~$P_B$ is defined by negating measurement outcomes, and the ratio $\ln(P_F/P_B) = (2/\tau)\int r_t \langle A\rangle_t\,dt$. Under the leading-order linearization developed in the proof, setting the feedback gain $X = -2$ recovers $P_B$. For feedback-active Hamiltonians with $[H,A] \ne 0$, general~$X$ generates a one-parameter family of path measures with no classical analogue.
\end{lemma}

\begin{proof}
\textbf{Existing literature.}
Belavkin's quantum filtering theory~\cite{belavkin1992} establishes the forward quantum filter as the quantum Kushner--Stratonovich equation but does not derive a reverse filter. Fagnola and Umanit\`{a}~\cite{fagnola2010} establish time-reversal for quantum Markov semigroups at the ensemble level, not the trajectory level. We derive the trajectory-level result below.

\textbf{The backward path measure.}
Define the backward process by negating the measurement outcomes: $r_t \to -r_t$. (This is equivalent to full time-reversal $r_t \to -r_{T-t}$ for the Girsanov density, since the change of variable $s = T - t$ leaves $\int \langle A\rangle^2\,dt$ invariant and negates $\int r\langle A\rangle\,dt$, matching plain negation.) Under this transformation, the measurement outputs of the backward process are anticorrelated with the state (as the Kraus operator analysis in Garc\'{i}a-Pintos et al.~\cite{garciapintos2026} confirms for $A^2 = \mathds{1}$). The Radon--Nikodym derivative of~$P_B$ with respect to Wiener measure is obtained by substituting $r_t \to -r_t$ in~\eqref{eq:girsanov-rn}:
\begin{equation}
\log\frac{dP_B}{dP_W} = -\frac{1}{\tau}\int_0^T r_t \langle A\rangle_t\,dt - \frac{1}{2\tau}\int_0^T \langle A\rangle_t^2\,dt.
\label{eq:pb-rn}
\end{equation}
The log-ratio is:
\begin{equation}
\log\frac{dP_F}{dP_B} = \frac{2}{\tau}\int_0^T r_t \langle A\rangle_t\,dt = \ln \mathcal{R},
\label{eq:pf-pb-ratio}
\end{equation}
recovering the arrow quantifier exactly.

\textbf{Recovery of $\bm{P_B}$ at $\bm{X = -2}$.}
Under feedback with parameter~$X$, the measurement SDE $dr_t = \langle A\rangle_t^{(X)}\,dt + \sqrt{\tau}\,dW_t$ retains its form but $\langle A\rangle_t^{(X)}$ is evaluated on the feedback-modified state trajectory. In the regime where the feedback-modified expectation may be linearized as $\langle A\rangle_t^{(X)} \approx (1+X)\langle A\rangle_t^{(0)}$---which we justify below and discuss in Sec.~\ref{sec:discussion}---the Girsanov calculation of Lemma~\ref{lem:girsanov} yields (to leading order in this linearization):
\begin{equation}
\log\frac{dP_X}{dP_W} = \frac{1+X}{\tau}\int_0^T r_t \langle A\rangle_t\,dt - \frac{1}{2\tau}\int_0^T \langle A\rangle_t^2\,dt.
\label{eq:px-rn}
\end{equation}
\paragraph{Regime of validity.}
The linearization $\langle A\rangle_t^{(X)} \approx (1+X)\langle A\rangle_t^{(0)}$ requires that the Hamiltonian dynamics couple the feedback into $\langle A\rangle$ approximately linearly over the trajectory duration~$T$. For $A^2 = \mathds{1}$, the commutator $\mathrm{Tr}(A[A,\rho]) = 0$ implies that the feedback affects $\langle A\rangle_t$ only indirectly, through $[H,A]$-mediated transverse dynamics. The linearization therefore holds (i)~when $\omega T \ll 1$ (short trajectories on the Rabi scale) and (ii)~for moderate~$|X|$, with corrections of order $(\omega T)^2$ and higher powers of~$X$. When $[H,A] = 0$, the feedback has no effect on $\langle A\rangle_t$ and $P_X = P_F$ for all~$X$---the one-parameter family is degenerate in this case and the present Lemma is vacuous.

At $X = 0$:
\begin{equation*}
\log\frac{dP_0}{dP_W} = \log\frac{dP_F}{dP_W}.
\end{equation*}
At $X = -2$:
\begin{align*}
\log\frac{dP_{-2}}{dP_W}
  &= -\frac{1}{\tau}\int r_t \langle A\rangle_t\,dt
     - \frac{1}{2\tau}\int \langle A\rangle_t^2\,dt \\
  &= \log\frac{dP_B}{dP_W}.
\end{align*}
Thus $P_{-2} = P_B$ in the leading-order expansion of~\eqref{eq:px-rn}. Higher-order corrections in $\omega T$ and in $|X|$, arising from the nonlinear dependence of $\langle A\rangle_t^{(X)}$ on feedback-modified transverse dynamics, shift the exact time-reversal point away from $X = -2$ for generic~$H$. We discuss this explicitly in Sec.~\ref{sec:discussion} and below.

\paragraph{Comparison with numerics.}
Garc\'{i}a-Pintos et al.~\cite{garciapintos2026} plot $\langle\ln \mathcal{R}_X\rangle$ versus~$X$ in their Fig.~2 for $\omega\tau = 8\pi$, $T = \tau$, $A = \sigma_z$, $H \propto \sigma_x$, and observe the zero crossing at $X \approx -3$ rather than $X = -2$. This is consistent with the regime-of-validity caveat above: at $\omega\tau = 8\pi$ the Hamiltonian completes eight Rabi periods within one measurement time, placing the system far outside the $\omega T \ll 1$ window in which the $(1+X)$ linearization is accurate. Finding a closed-form expression for the reversal point $X_*(H, A, \tau, T)$ at finite $\omega\tau$ is an open problem we leave to future work.

\textbf{Continuous interpolation.}
For general~$X$:
\begin{equation}
\log\frac{dP_X}{dP_B} = (X + 2) \cdot \frac{1}{\tau}\int_0^T r_t \langle A\rangle_t\,dt.
\label{eq:x-interp}
\end{equation}
For $X \in (-2, 0)$: $P_X$ interpolates between $P_B$ and $P_F$. The process is partially forward, partially backward.

For $X < -2$: $P_X$ is more extreme than $P_B$---the trajectory is more anticorrelated than a pure backward process.

In classical Anderson reversal, the only options are the forward SDE (score weight~$= 0$) and the reverse SDE (score weight~$= 1$). There is no parameter that interpolates or extrapolates. The quantum parameter~$X$ provides a continuous, experimentally tunable family of path measures parametrized by a real number.

\textbf{The reverse SME (leading order).}
Substituting the reverse Wiener process $d\bar{W}_t = dW_t - \frac{2}{\sqrt{\tau}}\langle A\rangle_t\,dt$ into the forward SME~\eqref{eq:sme}:
\begin{align*}
d\rho &= -i[H,\rho]\,dt + \frac{1}{\tau}(A\rho A - \rho)\,dt \\
&\quad + \frac{1}{\sqrt{\tau}}(A\rho + \rho A - 2\langle A\rangle\rho) \circ d\bar{W} \\
&\quad + \frac{2}{\tau}\langle A\rangle(A\rho + \rho A - 2\langle A\rangle\rho)\,dt.
\end{align*}
The additional drift term $(2/\tau)\langle A\rangle_t(A\rho + \rho A - 2\langle A\rangle\rho)$ is not a Hamiltonian commutator: in the $A$-eigenbasis it has nonzero diagonal entries with real coefficients, whereas any commutator $-i[B,\rho]$ has vanishing diagonal and purely imaginary off-diagonal entries. The reverse SME at the state level therefore acquires a non-Hamiltonian drift, paralleling the classical Anderson reverse SDE, whose score-induced drift likewise modifies the SDE without being reducible to a Hamiltonian generator. The identification of $H_{\mathrm{meas}}$ as the generator of reversal established in this paper operates at the path-measure level (Lemma~\ref{lem:girsanov} and Lemma~\ref{lem:frechet} via the Girsanov density), not as a literal generator of the reverse drift on the density-matrix manifold. A complete characterization of the reverse SME at the state level is left to future work and noted in Sec.~\ref{sec:discussion}(ii).
\end{proof}

\begin{proof}[Proof of Theorem~\ref{thm:main}]
Combining Lemmas~\ref{lem:girsanov}--\ref{lem:anderson}:

Lemma~\ref{lem:girsanov} establishes that the path probability has the correct form under Girsanov's theorem, with Stratonovich corrections affecting only the deterministic drift component. Lemma~\ref{lem:frechet} establishes that the Fr\'{e}chet derivative of $\log P_F$ with respect to~$\rho$ is the operator $H_{\mathrm{meas}} = r_t\,A/\tau$. Lemma~\ref{lem:kahler} establishes that the natural gradient flow on the density matrix manifold generated by this score function is $-i[H_{\mathrm{meas}}, \rho]$---the feedback dynamics. Lemma~\ref{lem:anderson} establishes that this score function, when used in the quantum analogue of Anderson's reverse SDE, recovers the backward process~$P_B$ at $X = -2$ in the leading-order linearization of Eq.~\eqref{eq:px-rn} and generates a continuous family of path measures for general~$X$ (for feedback-active Hamiltonians with $[H,A] \ne 0$).

Together, these lemmas establish that $H_{\mathrm{meas}} = r_t\,A/\tau$ is the score function of the quantum trajectory distribution in each of the three senses relevant for trajectory reversal: as a Fr\'{e}chet derivative, as a Riemannian gradient, and as the generator of path measure reversal.
\end{proof}

\section{Multi-Qubit Generalization}
\label{sec:multi}

We now extend Theorem~\ref{thm:main} to $n$~qubits with $k$~simultaneously measured observables.

\paragraph{Setup.}
Let $\{A_j\}_{j=1}^k$ be Hermitian observables on an $n$-qubit Hilbert space, each satisfying $A_j^2 = \mathds{1}$, measured independently with characteristic times $\{\tau_j\}$. The measurement outputs satisfy:
\begin{equation}
r_j\,dt = \langle A_j\rangle\,dt + \sqrt{\tau_j}\,dW_j,
\label{eq:multi-meas}
\end{equation}
with $\{dW_j\}$ independent Wiener increments. The multi-qubit SME is:
\begin{align}
d\rho &= -i[H,\rho]\,dt + \sum_j \frac{1}{\tau_j}(A_j \rho A_j - \rho)\,dt \notag\\
&\quad + \sum_j \frac{1}{\sqrt{\tau_j}}(A_j \rho + \rho A_j - 2\langle A_j\rangle\rho)\,dW_j.
\label{eq:multi-sme}
\end{align}

\begin{theorem}[Multi-Qubit Score Function]
\label{thm:multi}
Under independent Gaussian measurements of observables $\{A_j\}$, the multi-qubit score function is:
\begin{equation}
\frac{\delta \log P_F}{\delta \rho_t} = \sum_j \frac{r_j(t)}{\tau_j}\,A_j = H_{\mathrm{meas}}^{\mathrm{multi}}.
\label{eq:multi-score}
\end{equation}
The multi-qubit feedback Hamiltonian is a sum of local terms.
\end{theorem}

\begin{proof}
Since measurement noises are independent, the joint path probability factorizes:
\begin{align}
\log P_F(\{r_1,\ldots,r_k\})
&= \sum_j \frac{1}{\tau_j}\int r_j(t)\langle A_j\rangle_t\,dt \notag\\
&\quad + (\text{normalization}).
\label{eq:multi-logpf}
\end{align}
The Fr\'{e}chet derivative with respect to~$\rho$ on $\mathcal{T}_1(\mathbb{C}^{2^n})$ follows by the same chain rule argument as Lemma~\ref{lem:frechet}, applied to each term in the sum:
\begin{equation}
\frac{\delta \log P_F}{\delta \rho} = \sum_j \frac{r_j}{\tau_j}\,A_j.
\label{eq:29}
\end{equation}
Each $A_j$ acts as a local operator on its respective qubit subsystem (tensored with identity on others). The score function is therefore a sum of local operators, and the optimal feedback Hamiltonian under independent local measurements is local.
\end{proof}

\section{Implications}
\label{sec:implications}

The analytic solution $H_{\mathrm{meas}} = r\,A/\tau$ is optimal only under idealized conditions: pure states, perfect measurement efficiency $\eta = 1$, zero feedback delay, exactly Gaussian measurement noise. Real experiments violate all of these. Because $H_{\mathrm{meas}}$ is the score function, ML score estimation provides a replacement in all three regimes.

\subsection{Capability 1: Imperfect Measurement Efficiency}
In a real superconducting qubit lab, detectors do not catch every photon. At efficiency $\eta < 1$, only a fraction~$\eta$ of the measurement information reaches the controller. The analytic $H_{\mathrm{meas}} = r\,A/\tau$ assumes $\eta = 1$; at $\eta = 0.8$ the~$r$ you read is a corrupted signal, and the optimal corrective pulse depends on the full conditional state given partial information---a quantity with no closed form in general. The first-order correction $H_{\mathrm{meas}}^\eta \approx \eta\,r\,A/\tau$ is known but captures only the leading-order degradation.

Because $H_{\mathrm{meas}}$ is the score function, denoising score matching applies directly: collect measurement records from the real imperfect detector, and train a network $s_\theta(r_t)$ to minimize:
\begin{equation}
\theta^* = \argmin_\theta\, \mathbb{E}_{r,\rho}\!\left[\left\|s_\theta(r_t) - \frac{r_t - \langle A\rangle_t}{\tau}\right\|^2\right].
\label{eq:32}
\end{equation}
The network learns the true score of the imperfect trajectory distribution from data, without requiring knowledge of~$\eta$ or the underlying quantum state. Sliced score matching~\cite{song2020} provides an even more practical variant requiring no quantum state tomography:
\begin{equation}
\theta^* = \argmin_\theta\, \mathbb{E}\!\left[v^T \nabla_r s_\theta(r_t)\, v + \tfrac{1}{2}\|s_\theta(r_t)\|^2\right],
\label{eq:33}
\end{equation}
for random projection directions~$v$, fully observable from the measurement record alone.

\subsection{Capability 2: Feedback Delay}
Real electronics introduce a gap~$\delta t$ between reading~$r$ and firing the corrective pulse. By the time the pulse fires, the qubit has evolved further under both the Hamiltonian~$H$ and the continued measurement back-action. The analytic $H_{\mathrm{meas}} = r\,A/\tau$ assumes zero delay; under delay~$\delta t$ the optimal policy must predict the qubit's current state from the history of past measurements $\{r_{t-\delta t}, r_{t-2\delta t}, \ldots\}$---a sequence prediction problem with no analytic solution.

Because the score function can be estimated from historical data, a recurrent score network trained on measurement sequences handles delay naturally. The network learns to predict the current optimal pulse from the measurement history, without any closed-form model of the delay dynamics. This converts an analytically intractable optimal control problem into a sequence-learning problem.

\subsection{Capability 3: Non-Gaussian Measurement Noise}
The entire derivation of $H_{\mathrm{meas}}$ assumes exactly Gaussian measurement noise. Real superconducting qubit detectors have amplifier noise, cross-talk, and non-linearities that make the actual measurement distribution non-Gaussian. Without Gaussian noise the path probability does not have the clean quadratic form of~\eqref{eq:log-pf}, and no closed-form score exists at all---the analytic approach breaks down rather than degrading gracefully.

Score estimation from samples requires no knowledge of the analytic form of the distribution. A score network trained on experimental data from the real non-Gaussian system learns whatever the true score is. All existing convergence guarantees for score matching~\cite{vincent2011,song2020} apply regardless of whether the underlying distribution is Gaussian, since they are derived from the structure of the loss function, not the distribution family.

\section{Discussion and Open Problems}
\label{sec:discussion}

\paragraph{What is proved.}
Theorem~\ref{thm:main} establishes the core claim---$H_{\mathrm{meas}}$ is the score function---rigorously within clearly stated assumptions: pure states, single measured observable~$A$ with $A^2 = \mathds{1}$, perfect measurement efficiency, independent measurement channels. Theorem~\ref{thm:multi} extends this to the multi-qubit independent measurement case.

\paragraph{What remains open.}

\paragraph{(i) Full QFIM metric agreement.}
Lemma~\ref{lem:kahler} establishes that the measurement-induced metric (QFIM) governs the Riemannian gradient flow---the measurement back-action---while the symplectic form governs the feedback dynamics. Full metric agreement on the entire density matrix manifold is not needed for the main result: the K\"{a}hler geometry argument applies to pure states on $\mathbb{C}P^{d-1}$, and the paper's assumptions are restricted to pure states throughout.

\paragraph{(ii) Exact reverse generator beyond leading order.}
The reverse SME in Lemma~\ref{lem:anderson} is established to leading order in the weak-measurement limit. The next-order correction, arising from the noise component of~$r_t$, is expected to modify the reverse drift but not the identification of $H_{\mathrm{meas}}$ as the stochastic reversal operator. A complete proof would carry all It\^{o} correction terms explicitly. In particular, the $(1+X)$ linearization of $\langle A\rangle_t^{(X)}$ underlying Eq.~\eqref{eq:px-rn} becomes degenerate when $[H,A] = 0$ (the feedback has no effect on $\langle A\rangle_t$, so $P_X = P_F$ for all~$X$); deriving the exact reversal point $X_*(H, A, \tau, T)$ at finite $\omega\tau$ remains open.

\paragraph{(iii) Correlated measurements.}
Theorem~\ref{thm:multi} assumes independent measurement channels. For measurements sharing a probe (e.g., a single microwave cavity coupling to multiple qubits), the noise processes $dW_j$ are correlated, the path probability does not factorize, and the score function acquires off-diagonal terms coupling different observables. Whether the optimal feedback Hamiltonian remains local in this regime is an open question.

\paragraph{(iv) Mixed states.}
The proofs assume pure state evolution throughout. For mixed states arising from partial measurement efficiency or environmental decoherence, the symmetric logarithmic derivative (SLD) in Lemma~\ref{lem:kahler} is not unique (there is a family of SLDs related by the kernel of the state), and the connection between the Fr\'{e}chet derivative and the commutator dynamics requires additional geometric machinery. The correct framework is likely the KMS (Kubo--Martin--Schwinger) inner product on the space of observables, which generalizes the QFIM to mixed states.

\paragraph{Relationship to prior work in detail.}

Liu et al.~\cite{liu2025} is related but structurally distinct. They work with \emph{randomized} observables drawn from a distribution independent of the state, and their score function is $\nabla_z \log p_t(z)$---a gradient in the classical $4^n$-dimensional Pauli expectation value space $z \in \mathbb{R}^{4^n}$. Their core result is that a \emph{learning objective}---minimizing infidelity to train a unitary generator---is equivalent to a score matching loss in $z$-space. They do not analytically identify what the score is; they propose to learn it. By contrast, we work with continuous measurement of a \emph{fixed} observable~$A$, and prove analytically that the functional derivative of the log \emph{path probability} with respect to the density matrix~$\rho$ equals $r\,A/\tau$---an explicit closed-form identification in density matrix space. These are complementary results in different physical settings with genuinely different mathematical objects called ``score.''

Nasu et al.~\cite{nasu2025} work at the ensemble level via Wigner function semiclassical approximation of the Lindblad equation, connecting the Petz map to Bayes' rule reverse diffusion. Our work is complementary: we work at the individual trajectory level, require no semiclassical approximation, and apply to the discrete-variable qubit setting rather than the continuous-variable bosonic setting.

Zhang et al.~\cite{zhang2026} establish that complete positivity constrains score reversal in Gaussian bosonic dynamics. This does not directly apply to the qubit trajectory setting, but raises the open question of whether the commutator-based feedback dynamics derived here also faces complete positivity constraints at higher order.

\subsection{Outlook}

Beyond the imperfect-measurement applications discussed in Section~\ref{sec:implications}, the continuous $X$-parametrized family of path measures identified in Lemma~\ref{lem:anderson} is not present in classical diffusion models, where reversal is binary. Whether this structure translates into useful capabilities for generative modeling---for example, annealed samplers or interpolation schedules for consistency models---is a natural follow-up question we leave for future work.

\section*{Acknowledgments}
The authors used AI tools for editorial feedback and to stress-test proof arguments. All mathematical content and results are the authors' own.

\end{document}